\documentclass[runningheads]{llncs}

\usepackage[utf8]{inputenc}
\usepackage{amsmath}

\usepackage{amsthm}
\usepackage{stmaryrd}
\usepackage{bbold}
\usepackage{graphicx}
\usepackage{subcaption}
\usepackage[linesnumbered,ruled,vlined,noend]{algorithm2e}
\usepackage{thmtools}
\usepackage{thm-restate}
\usepackage{url}
\usepackage{enumitem}
\usepackage[T1]{fontenc}

\def\N{\mathbb{N}}
\def\P{\mathbb{P}}

\def\E{\mathbb{E}}
\def\B{\mathcal{B}}

\def\O{\mathcal{O}}

\def\1{\mathbb{1}}

\def\Poi{\operatorname{Poi}}
\def\Ber{\operatorname{Ber}}
\def\Var{\operatorname{Var}}
\def\AE{\operatorname{AE}}

\newcommand{\ceil}[1]{\left\lceil#1\right\rceil}
\newcommand{\floor}[1]{\left\lfloor#1\right\rfloor}

\newcommand{\ocinterval}[1]{\left(#1\right]}
\newcommand{\cointerval}[1]{\left[#1\right)}

\usepackage{xcolor}

\begin{document}


\title{A Poisson-Based Approximation Algorithm for Stochastic Bin Packing of Bernoulli Items}
\titlerunning{A Poisson-Based Approximation Algorithm for Stochastic Bin Packing}

\author{Tomasz Kanas\orcidID{0000-0002-0715-4622} \and
  Krzysztof Rzadca\orcidID{0000-0002-4176-853X}}

\institute{Institute of Informatics, University of Warsaw, Poland\\
  \email{\{t.kanas,krzadca\}@mimuw.edu.pl}}

\maketitle

\begin{abstract}
  A cloud scheduler packs tasks onto machines with contradictory goals of (1) using the machines as efficiently as possible while (2) avoiding overloading that might
  result in CPU throttling or out-of-memory errors.
  We take a stochastic approach that models the uncertainty of tasks' resource requirements by random variables.
  We focus on a little-explored case of items, each having a Bernoulli distribution that corresponds to tasks that are either idle or need a certain CPU share.
  RPAP, our online approximation algorithm, upper-bounds a subset of items by Poisson distributions.
  Unlike existing algorithms for Bernoulli items that prove the approximation ratio only up to a multiplicative constant, we provide a closed-form expression.
  We derive RPAPC, a combined approach having the same theoretical guarantees as RPAP.
  In simulations, RPAPC's results are close to FFR, a greedy heuristic with no worst-case guarantees; RPAPC slightly outperforms FFR on datasets with small items.
  

\end{abstract}


\section{Introduction}

Modern virtualization technologies --- virtual machines (VMs) and Linux containers --- allow execution in parallel of dozens of independent tasks on a single physical machine. 
Given the planet-wide scale~\cite{barroso2013datacenter} of the largest public (AWS, Azure, GCP) and private (e.g. Google) clouds, 
even small improvements in resource utilization slow the growth rate of the hardware fleets and thus save equipment and electricity~\cite{barroso2013datacenter,Autopilot,bashir2021peaklimit}.

Bin Packing (BP)~\cite{johnson1973near} is 
perhaps the most fundamental model of datacenter allocation~\cite{verma2015large,Tirmazi_2020,pietri2016mapping}.
In BP, the goal is to pack the given items into as few equally-sized bins as only possible, without exceeding the capacity of any bin.
In cloud computing, 
bins correspond to machines, items to 
\emph{tasks} (VMs or containers) to allocate and items'
sizes --- to CPU or memory requirements.

However, there is a fundamental difference between packing boxes onto a truck and Linux containers onto a machine. 
Boxes' sizes are easy to measure and, barring extreme events, unchanging. In contrast, the resource requirements of a task are more difficult to estimate. 
Tasks are commonly packed by \emph{limits}~\cite{verma2015large,Tirmazi_2020,burns2016borg}: essentially, the to-be-scheduled task declares (sometimes through automation~\cite{Autopilot}) to the scheduler an upper bound on the resources it might request. 
Yet, packing by limits is fundamentally inefficient~\cite{bashir2021peaklimit}. 
Even if limits were clairvoyant (set to each task's exact maximal usage), 
using limits, 
the scheduler effectively assumes that every task will always consume exactly its maximal usage --- which is rarely the case~\cite{Janus_2017}. Even with overcommit~\cite{bashir2021peaklimit}, utilization remains low~\cite{lu2017imbalance}.




Stochastic Bin Packing (SBP)~\cite{Goel_1999,bursty} 
models the uncertainty of tasks' resource requirements by using random variables as items' sizes.
Accordingly, the constraint of never overpacking any bin is generalized to a probabilistic one --- an upper bound $\alpha$ on the probability that each bin's capacity is exceeded.
SBP can represent the cloud allocation problem~\cite{Janus_2017,Breitgand_2012,Cohen_2017}:\@ the random variables map to observed or estimated tasks' resource usage; and $\alpha$ maps to a probabilistic Service Level Objective, SLO\@.
Notably,~\cite{bashir2021peaklimit} combines declared limits (for new tasks) with estimations of a machine's predicted total usage (for long-running ones); a prototype improved efficiency by 2\% on 11,000 production machines in the internal Google cloud.
While SBP models have limitations (e.g.: not explicitly modeling variability over time~\cite{Luo_2013}, dynamic arrivals and departures~\cite{coffman1983dynamic}, or correlations between tasks~\cite{beaumont2016analyzing}),
we claim that solving a more general problem usually requires at some point solving its more fundamental version.

Perhaps the most restrictive assumption we take is that all the items follow scaled Bernoulli distributions. 
Such items correspond to tasks that for some fraction of time compute with (approximately) constant intensity, and then idle e.g.\ waiting for the next request. 
We claim that Bernoulli items are a reasonable model: e.g., in the Google Cluster Trace~\cite{Wilkes2020a}, \cite{Janus_2017}~shows a large task group with CPU requirements resembling the scaled Bernoulli distribution. 
One can argue that if there were enough tasks in one bin, then, from the Central Limit Theorem, the cumulative distribution of that bin would be close to normal.
However, if the tasks are large, few of them fit into a machine, which makes the normal distribution inadequate~\cite{Janus_2017}.
Additionally, solving the special case of Bernoulli items could bring us closer to a distributionally-robust solution.
From the theoretical perspective, Bernoulli items seem to pose more difficulties than other distributions like Poisson~\cite{Goel_1999} or Gaussian~\cite{Breitgand_2012,Cohen_2017} (Section~\ref{sec:related-work}).

\noindent\textbf{The contribution of this paper is the following:}
\begin{itemize}[nosep,noitemsep,topsep=0pt,parsep=0pt,partopsep=0pt,leftmargin=*]
  \item We design Refined Poisson Approximation Packing (RPAP), an online algorithm that finds a viable packing of Bernoulli variables to bins while keeping the overload probability of any bin below $\alpha$. 
  Our algorithm is easy to implement and schedules one item in $\O(\log n)$ time (Section~\ref{sec:rpap}).
  \item We prove a closed-form formula of the RPAP approximation ratio, which depends only on the (given) overload probability $\alpha$ (Section~\ref{sec:approximation}, Eq.~\ref{eq:approx_constant_general}).
  \item 
  In simulations, we compare RPAP with \cite{bursty} and FFR, a heuristic with no worst-case guarantees.
  We propose RPAPC that combines RPAP with a heuristic, maintaining RPAP's guarantees. Our approaches outperform \cite{bursty} and are close to FFR; slightly improving upon FFR on datasets with small items.
\end{itemize}
To the best of our knowledge, our paper shows the first proof of a closed-form formula for the approximation ratio of an algorithm for Stochastic Bin Packing with Bernoulli items (\cite{bursty} shows only asymptotics) and the first experimental evaluation of SBP algorithms on Bernoulli items.

%
\section{Related work}\label{sec:related-work}
We focus below on theoretical approaches to stochastic bin packing.
SBP is a stochastic extension of a classic combinatorial optimization problem, an approach called stochastic optimization~\cite{Goel_1999}.
Works on SBP usually assume that all items' sizes follow a known distribution.
When items have \emph{normal} distribution,
Breitgand and Epstein~\cite{Breitgand_2012} show a $(2+\epsilon)$-approximation algorithm, and an offline 2-approximation; Cohen et al.~\cite{Cohen_2017} show that First Fit is 9/4-approximation;
Martinovic and Selch~\cite{Martinovic_2021}
show improved lower bounds and discuss linearization techniques; Yan et al.~\cite{Yan_2022} propose a new metric of bin load, develop algorithms and perform experiments on synthetic and real data.
Other item distributions are also considered, for example, Goel and Indyk~\cite{Goel_1999} propose a PTAS for \emph{Poisson} and \emph{exponential} items.

The \emph{Bernoulli} distribution seems to be more difficult to work with. For a bin, computing the overflow probability is $\O(n)$ for Poisson and Gaussian distributions; yet it is \#P hard for Bernoulli~\cite{bursty}, i.e. as hard as counting the number of solutions of an NP-complete problem (which is hypothesized to be harder than finding any solution).
Furthermore, a standard approach to stochastic bin packing is to calculate each stochastic item's \emph{effective size}, which is then used by a deterministic packing algorithm.
For Poisson and normal items, one can find an effective size that gives an $\O(1)$-approximation algorithm~\cite{Chen_2011}.
However, for Bernoulli items, any effective size-based algorithm has an $\Omega(\alpha^{-1/2})$ approximation ratio, where $\alpha$ is the maximal overflow probability~\cite{bursty}.
\cite{Goel_1999} shows a QPTAS for \emph{Bernoulli} items. 
\cite{bursty} shows an $\O\left(\sqrt{\frac{\log \alpha^{-1}}{\log\log \alpha^{-1}}}\right)$-approximation and $\O(\epsilon^{-1})$-approximation for an $\epsilon$-relaxed problem. 

\noindent\textbf{Our approach compared to Kleinberg et al.~\cite{bursty}:} Like~\cite{bursty}, our algorithm also splits items into subgroups and similarly packs the small items (Section~\ref{sec:correctness-confident}).
In contrast, for the most complex case of the standard items we use a Poisson approximation (Section~\ref{sec:correctness-standard}), while they use effective bandwidth and probabilistic inequalities.
Moreover, there is only asymptotic analysis of the approximation factor in~\cite{bursty}, which allows them to hide in the $\O$ notation the multiplicative constant arising from splitting items into subgroups.
We managed to avoid such multiplicative constant by proving the upper bound of the expected value of any correctly packed bin (Lemma~\ref{upper_mean}).
Moreover, to bound the approximation constant, we proved a technical inequality on the inverse of Poisson CDF (Lemma~\ref{lem:inv_gamma_ineq}). In contrast, \cite {bursty} used results on antichains to optimize the asymptotic approximation factor of their algorithm.

\section{Problem formulation and notation}
We are given a sequence of items $X_{1},\ldots, X_{n}$ and an infinite sequence of identical bins of capacity 1. 
The goal is to find a \emph{viable} assignment of items to bins that uses the minimal number of bins. 
We assume that all items are random variables that follow \emph{scaled Bernoulli} distributions. As in~\cite{Goel_1999,bursty}, we assume that random variables are independent. Our problem is thus clairvoyant, as we receive full information about an item on submission, although the sizes remain stochastic, in contrast to an alternative model in which a size is drawn from a certain distribution and then does not change.

We denote the Bernoulli distribution by $\Ber(p)$ and the Poisson distribution by $\Poi(\lambda)$. We define \emph{scaled} Bernoulli $\Ber(p, s)$ and Poisson $\Poi(\lambda, s)$ distributions, where $s > 0$ is the \emph{size}: $sX$ is scaled-Bernoulli distributed ($sX \sim \Ber(p,s)$) when $X \sim \Ber(p)$ (Poisson is defined analogically). For example, if an item $X_{i} \sim Ber(p, \frac{1}{3})$ then the item's size is equal to $\frac{1}{3}$ with probability $p$ and $0$ with probability $1-p$.

For a random variable $X$, $F_{X}$ denotes its cumulative distribution function (CDF), $F_{X}(t) = \P(X \le t)$.
We denote $Q$ as the Poisson CDF: $\P(\Poi(\lambda) \le x) = Q(x, \lambda)$; and $Q^{-1}(x,\gamma)$ as its inverse with respect to the second argument.

We assume that items $X_{i} \sim \Ber(p_{i},s_{i})$, where $p_{i} \in \ocinterval{0, 1}$, and $s_{i} \le s_{\max}$. $s_{\max} \in \ocinterval{0,1}$ is an additional parameter that increases the versatility of our results.
In the general case, $s_{\max} = 1$ (an item always fits in a single bin).
We denote the set of items in $j$-th bin by $\B_{j}$, their sum by $B_{j}$ and by $\alpha > 0$ the \emph{maximal overflow probability}. An assignment is \emph{viable} if for every bin $j$ the probability of exceeding the bin's capacity is at most $\alpha$, $\P(\sum_{i \in \B_{j}} X_{i} > 1) \le \alpha$.

We argue that $\alpha$ should be treated as a constant in the context of the data center allocation,
where $\alpha$ corresponds to the service level objective (SLO) negotiated between the provider and their clients. 
As only very rarely can the machine be overloaded,
usually, there are only a few groups of items with fixed and small SLO values (e.g., $0.01, 0.005, 0.001$). We thus also assume that $0 < \alpha \le \frac{1}{2}$.

We call a BP algorithm Any-Fit if it does not open a new bin if the current item fits in any already opened bin~\cite{CoffmanJr_2013} (e.g. First Fit or Best Fit). We use Any-Fit algorithms as a building block for RPAP, but RPAP \emph{is not} Any-Fit.

\section{Refined Poisson Approximation Packing Algorithm}\label{sec:rpap}

Refined Poisson Approximation Packing (RPAP, Algorithm~\ref{algo}), separates items into three disjoint groups. 
Each group is packed separately into a disjoint set of bins.
We reduce the packing of each group to BP and pack with an Any-Fit algorithm.
In this section, we describe the algorithm; the following Section~\ref{sec:correctness} proves the viability of the allocation; and Section~\ref{sec:approximation} proves the approximation ratio.

To separate items into three groups, we introduce two additional parameters: $p_{\max} \in (0,1)$, $s_{\min} \in (0, s_{\max})$
(we show in Section~\ref{sec:approximation-opt} how to choose the values that minimize the approximation ratio).
The groups are defined as follows:

\begin{itemize}[nosep,noitemsep,topsep=0pt,parsep=0pt,partopsep=0pt]
  \item \emph{Confident} items have non-zero load with high probability: $p > p_{\max}$.
  \item \emph{Minor} items are small: $s \le s_{\min}$, $p \le p_{\max}$.
  \item \emph{Standard} items are the remaining items: $s_{\min} < s \le s_{\max}$, $p \le p_{\max}$.
\end{itemize}

The algorithm proceeds as follows. 
Confident items have large probabilities, so we round their probabilities up to 1 and pack them by their sizes (line~\ref{algo_confident}).
Minor items are small, so they have small variances because the variance of $X \sim \Ber(p, s)$ is $s^{2}p(1-p)$. Intuitively it means that with high probability small items are close to their mean.
Thus, we pack them (line~\ref{algo_minor}) by their means scaled by some factor $\mu_{0}$ (defined in line~\ref{algo_mu_0}).

The core idea of our algorithm is to approximate the remaining, \emph{standard}, items by Poisson variables. 
The problem of packing Poisson variables turns out to be equivalent to BP. We later prove that we can upper bound a $\Ber(p)$ variable by a $\Poi(\log(1/(1 - p)))$ variable.
As the items are \emph{scaled} Bernoulli variables, we also use scaled Poisson variables, but to reduce the problem to BP, we need these sizes to be equal.
Thus, we additionally group standard items into subgroups with similar sizes and round their sizes up to the upper bound of such subgroup (line~\ref{algo_standard}: $k$ is the subgroup and $\lambda_k$ scales all items' sizes in that group).

\begin{algorithm}[tb]
  \SetDataSty{text}
  \SetFuncSty{text}
  \SetAlgoLined{}
  \SetKwData{ConfidentBins}{ConfidentBins}
  \SetKwData{MinorBins}{MinorBins}
  \SetKwData{StandardBins}{StandardBins}
  \SetKwData{EmptyPacking}{EmptyPacking}
  \SetKwData{Packing}{Packing}
  \SetKwFunction{PackAnyFit}{PackAnyFit}
  \SetKwFunction{Concat}{Concat}
  \SetKwInOut{Data}{Data}
  \SetKwInOut{Using}{Using}
  \SetKwInOut{Return}{return}
  \SetKwFor{Foreach}{for each}{:}{}
  \Using{\PackAnyFit{id, size} method that packs item id with an Any-Fit algorithm to a bin of size 1.}
  \ConfidentBins$:=$ \EmptyPacking\;
  \MinorBins$:=$ \EmptyPacking\;
  $k_{\min} := \floor{1/s_{\max}}$\;\label{algo:k_min}
  $k_{\max} := \ceil{1/s_{\min}} - 1$\;\label{algo:k_max}
  \Foreach{$k \in \{k_{\min},\ldots,k_{\max}\}$}{
    $\lambda_{k} := Q^{-1}(k + 1, 1 - \alpha)$\;\label{algo:lambda_k}
    $\StandardBins[k]:=$ \EmptyPacking\;
  }
  $\mu_0 := (2\alpha + s_{\min} - \sqrt{s_{\min}^{2} + 4\alpha s_{\min}}) / 2 \alpha$\;\label{algo_mu_0}
  \Foreach{item i}{
      \uIf{$p_{i} > p_{\max}$}{
        $\ConfidentBins.\PackAnyFit(i,s_{i})$\;\label{algo_confident}
      } \uElseIf{$s_{i} \le s_{\min}$}{
        $\MinorBins.\PackAnyFit(i, p_{i}s_{i} / \mu_0)$\;\label{algo_minor}
      } \Else{
          $k := \floor{1/s_{i}}$\;\label{algo_standard_1}
          $\StandardBins[k].\PackAnyFit(i, \log(1/(1 - p_{i})) / \lambda_{k})$\;\label{algo_standard}
      }
  }
  \mbox{\Return{(\ConfidentBins, \MinorBins, $\StandardBins[k_{\min}], \dots, \StandardBins[k_{\max}]$)}}
  \caption{Refined Poisson Approximation Packing (RPAP)}\label{algo}
\end{algorithm}

\section{Proof of Correctness}\label{sec:correctness}
As RPAP packs the three groups into three disjoint sets of bins, we prove the correctness of the allocation case by case: confident and minor items in Section~\ref{sec:correctness-confident}; and standard items in Section~\ref{sec:correctness-standard}.

\subsection{Confident and minor items}\label{sec:correctness-confident}
\begin{lemma}
  The packing of confident and minor items is viable.
\end{lemma}
\begin{proof}
Confident items are packed by their sizes $s_{i}$, so the sum of sizes in any bin $\B$ is $\sum_{i} s_{i} \le 1$, and the probability of overflow is $\P(B > 1) = 0 < \alpha$.

For minor items we have $\forall_{i} s_{i} \le s_{\min}$ and we are packing them by their expected value $s_{i}p_{i}$, so if $\E(B) = \sum_{i}s_{i}p_{i} \le \mu_{0} < 1$, then from Chebyshev inequality:
\[\P(B > 1) \le \frac{\Var(B)}{{(1 - \E(B))}^{2}} = \frac{\sum_{i}s_{i}^{2}p_{i}(1-p_{i})}{{(1 - \E(B))}^{2}} < \frac{s_{\min}\sum_{i}s_{i}p_{i}}{{(1 - \E(B))}^{2}} \le s_{\min}\frac{\mu_{0}}{{(1 - \mu_{0})}^{2}}.\]

The viability of the packing follows from $\mu_0$ (Algorithm~\ref{algo}, line~\ref{algo_mu_0})
being a solution of the equation:
$\frac{{(1 - \mu_{0})}^{2}}{\mu_{0}} = \frac{s_{\min}}{\alpha}.$
\end{proof}

\subsection{Standard items}\label{sec:correctness-standard}

To pack standard items, we upper-bound the probability of overflow by the tail of the Poisson distribution. First, we separate items into subgroups, such that the $k$-th group consists of items whose sizes are in the interval $s_{i} \in \ocinterval{\frac{1}{k+1}, \frac{1}{k}},\quad k \in \{k_{\min}, \ldots,  k_{\max}\}$ (Algorithm~\ref{algo}, lines~\ref{algo:k_min}-\ref{algo:k_max}).
Inside a single subgroup, we round items' sizes up to the upper bound of the interval. Every subgroup is packed into a separate set of bins.

The proof uses the following two lemmas (all proofs are in the appendix~\cite{paper-appendix}):
\begin{restatable}{lemma}{berpoilemma}\label{lem_bern_poiss}
  \mbox{If $X \sim \Ber(p)$, $Y \sim \Poi(\lambda)$, and $\lambda \ge \ln\left(\frac{1}{1-p}\right)$, then
  $\forall_{t} F_{X}(t) \ge F_{Y}(t) \text{.}$}
\end{restatable}

\begin{restatable}{lemma}{sumtwolemma}\label{lem_cdf_sum}
  If $X_{1},X_{2},Y_{1},Y_{2}$ are discrete independent random variables with \linebreak countable support and ${\forall_{t} F_{X_{i}}(t) \ge F_{Y_{i}}(t)}$ for $i \in \{1,2\}$, then
  $\forall_{t} F_{X_{1} + X_{2}}(t) \ge F_{Y_{1} + Y_{2}}(t)$.
\end{restatable}

The following lemma shows that a viable packing of Poisson variables is also a viable packing of the original Bernoulli variables and is a direct consequence of the above lemmas. 

\begin{restatable}{lemma}{stochmajorizationlem}
  Let $X_{i} \sim \Ber(p_{i}, s_{i}),\ P_{i} \sim \Poi\left(\ln\left(\frac{1}{1-p_{i}}\right), \bar s_{i}\right)$ for $i \in \{1,\ldots, m\}$ be independent and $\forall_{i}\bar s_{i} \ge s_{i}$. Moreover let $P = \sum_{i=1}^{m}P_{i}$ and $B = \sum_{i=1}^{m}X_{i}$. Then $\P(B > 1) \le \P(P > 1)$.
\end{restatable}

The following lemma shows that BP of scaled Poisson variables by their means is viable: 

\begin{lemma}
  Packing of scaled Poisson variables $P_{i} \sim \Poi(\lambda_{i}, s)$ is viable if and only if it is  correct packing of their means $\lambda_{i}$ with bin size $Q^{-1}(\floor{\frac{1}{s}} + 1, 1-\alpha)$, i.e.\ for every $\B_{j}$:
    $\P\left(\sum_{i \in \B_j}P_{i} > 1\right) \le \alpha \iff \sum_{i \in \B_{j}}\lambda_{i} \le Q^{-1}\left(\floor{\frac{1}{s}} + 1, 1-\alpha\right)$.
\end{lemma}
\begin{proof}
  As variables $P_{i}$ are independent, the load of any bin $P = \sum_{i\in \B}P_{i}$ is also a scaled Poisson variable $P \sim \Poi(\lambda, s)$, where $\lambda = \sum_{i\in \B} \lambda_{i}$. We have $\frac{1}{s}P \sim \Poi(\lambda)$, so
    $\P(P \le 1) = \P\left(\frac{1}{s}P \le \floor{\frac{1}{s}}\right) = Q\left(\floor{\frac{1}{s}} + 1, \lambda\right)$
  and the thesis follows from applying $Q^{-1}\left(\floor{\frac{1}{s}} + 1, \cdot \right)$. 
\end{proof}

\section{Approximation Ratio}\label{sec:approximation}
We start in Section~\ref{sec:approximation-proof} by a formula for the approximation ratio of RPAP\@. 
Then, in Section~\ref{sec:approximation-opt}, we optimize the approximation ratio by adjusting parameters: the least-probable confident item $p_{\max}$ and the largest minor item $s_{\min}$.

\subsection{Proof of the Approximation Ratio}\label{sec:approximation-proof}
We prove the approximation factor by lower bounding the expected value of an average bin for any packing produced by RPAP, and upper bounding this average bin expected value for any viable packing, in particular the optimal one.

We will prove the upper bound by induction over the number of items in a bin.
We need a stronger induction assumption: 
we want to reward adding items that increase the expected load of a bin \emph{without increasing the overflow probability}.
We model that by introducing a \emph{discount function}: $C(X) = \sum_{x \in [0,1]}(1 - x)\P(X = x)$ ($X$ has finite support, so the sum is well-defined).


\begin{lemma}\label{upper_mean}
  Assume that $\forall_{i \in \B}\: X_{i} \sim \Ber(p_{i},s_{i})$, $s_{i} \le 1$ are independent random variables. If $B = \sum_{i \in \B}X_{i}$, satisfies
  $\P(B > 1) \le \alpha < 1$, 
  then $\E(B) \le \frac{1 + \alpha}{1 - \alpha}$.
\end{lemma}
\begin{proof}
  Without loss of generality, assume that $\B = \{1,\ldots, N\}$, and denote $S_{n} = \sum_{i=1}^{n}X_{i}$.
  We proceed by induction over $n$, with the following assumption:
  \[\E(S_{n}) \le \frac{1}{1 - \alpha}\Bigg(1 + \P(S_{n} > 1) - C(S_{n})\Bigg).\]
  Notice that it is enough to prove the induction, as $C(S_{n}) \ge 0$.

  {\bf Basis of induction:} for $n = 0$ we have $\P(S_{0} = 0) = 1$ so $\P(S_{0} > 1) = 0$, $C(S_{0}) = 1$ and the thesis follows.

  {\bf Inductive step:} Let us denote $S_{n+1} = S_{n} + X$, $X \sim \Ber(p,s)$.
  From the induction assumption
  \[\E(S_{n+1}) = \E(S_{n}) + ps \le \frac{1}{1-\alpha}\Bigg(1 + \P(S_{n} > 1) - C(S_{n})\Bigg) + ps,\]
  so it suffices to show that
  \[\P(S_{n} > 1) - C(S_{n}) + (1-\alpha)ps \le \P(S_{n+1} > 1) - C(S_{n+1}).\]
  We have recursive formulas:
  \[\P(S_{n+1} > 1) = \P(S_{n} + X > 1) = \P(S_{n} > 1) + p\P(S_{n} \in \ocinterval{1-s, 1}) \text{;}\]
  \begin{equation*}
    \begin{aligned}
      C&(S_{n+1}) 
      = (1-p)C(S_{n}) + p\sum_{x \in [0, 1]}(1-x)\P(S_{n} = x - s) =\\
       &= C(S_{n}) - p\sum_{x \in \ocinterval{1-s, 1}}(1-x)\P(S_{n} = x) - ps\P(S_{n} \in [0, 1-s]).
    \end{aligned}
  \end{equation*}
  So after simplifications, we arrive at the inequality:
  \begin{equation*}
    \begin{aligned}
      (1-\alpha)s& \le \P(S_{n} \in \ocinterval{1-s, 1}) + s\P(S_{n} \in [0, 1-s]) + \\
      &+ \sum_{x \in \ocinterval{1-s, 1}}(1-x)\P(S_{n} = x)  = s\P(S_{n} \in [0, 1]) + A
    \end{aligned}
  \end{equation*}
  Where
  \[A = (1-s)\P(S_{n} \in \ocinterval{1-s,1}) + \sum_{x \in \ocinterval{1-s, 1}}(1-x)\P(S_{n} = x) \ge 0\]
  what completes the induction step, as $\P(S_{n} \in [0,1]) \ge 1-\alpha$.
\end{proof}

Next, we will prove the lower bound on the average expected value of all bins in packing produced by RPAP\@.
Recall that in all item groups, we used at some point an Any-Fit algorithm,
so we will need this slightly stronger version of the classic lemma~\cite{CoffmanJr_2013}:

\begin{lemma}\label{lem_greedy}
  Assume that an Any-Fit algorithm packed real values $x_{1},\ldots,x_{n}$ into bins $\B_{1},\ldots,\B_{m}$ where $m \ge 2$. Then $\frac{1}{m}\sum_{j=1}^{m}\sum_{i \in \B_{j}} x_{i} > \frac{1}{2}$
\end{lemma}
\begin{proof}
  Denote $B_{j} := \sum_{i \in \B_{j}} x_{i}$. An Any-Fit algorithm opens a new bin only if the current item does not fit into any already open bin, so
  $\forall_{j}\forall_{l \neq j} B_{j}+B_{l} > 1$.
  If $m$ is even then
  $\sum_{j=1}^{m}B_{j} > \frac{m}{2},$
  and the lemma is proved. If $m$ is odd then
  \[2\sum_{j=1}^{m}B_{j} = \sum_{j=1}^{m-1}B_{j} + \sum_{j=2}^{m}B_{j} + B_{1} + B_{m} > 2\frac{m-1}{2} + 1 = m.\]
\end{proof}

As the lemma above does not hold for the special case with a single bin,
we proceed with the proof for the typical case of at least two bins
and deal with the special case directly in the proof of Theorem~\ref{th:approx-ratio}.
We denote the average expected value of bins $\B_{1}, \ldots, \B_{m}$ by
$\AE(\B) = \frac{1}{m}\sum_{j=1}^{m}\sum_{i \in \B_{j}} p_{i}s_{i}$.
The following three lemmas are very similar and follow easily from the Lemma~\ref{lem_greedy}, so we will prove only the last (most complex) one.
\begin{lemma}\label{lem:confident_mean}
  If $\B_{1},\ldots,\B_{m}$, $m \ge 2$ are bins with the confident items, then their average expected value fulfills $\AE(\B) > \frac{p_{\max}}{2}$.
\end{lemma}

\begin{lemma}\label{lem:minor_mean}
  If $\B_{1},\ldots,\B_{m}$, $m \ge 2$ are bins with the minor items, then their average expected value fulfills $\AE(\B) > \frac{\mu_{0}}{2}$.
\end{lemma}

\begin{lemma}\label{lem:standard_mean}
  If $\B_{1},\ldots,\B_{m}$, $m \ge 2$ are bins with the standard items of the $k$-th subgroup, then their average expected value fulfills $\AE(\B) > \frac{\lambda_{k}(1-p_{\max})}{2(k + 1)}$.
\end{lemma}
\begin{proof}
We packed the standard items of the $k$-th subgroup by $\frac{1}{\lambda_{k}}\ln\left(\frac{1}{1-p_{i}}\right)$ (Algorithm~\ref{algo}, line~\ref{algo_standard}). For any such item $X_{i} \sim \Ber(p_{i},s_{i})$ we have $s_{i} > \frac{1}{k+1}$
and
$\ln\left(\frac{1}{1-p_{i}}\right) \le \frac{p_{i}}{1-p_{i}} \le \frac{k + 1}{1 - p_{\max}}p_{i}s_{i}.$
So from Lemma~\ref{lem_greedy}:
\[\frac{m}{2} < \frac{1}{\lambda_{k}}\sum_{j=1}^{m}\sum_{i \in \B_{j}}\ln\left(\frac{1}{1-p_{i}}\right) \le \frac{k + 1}{\lambda_{k}(1 - p_{\max})} \sum_{j=1}^{m}\sum_{i \in \B_{j}}p_{i}s_{i}.\]
\end{proof}

Additionally, we need the following result (proof in~\cite{paper-appendix}) to find $k$ for which $\frac{\lambda_{k}}{k+1}$ is minimal:
\begin{restatable}{lemma}{invgammaineqlem}\label{lem:inv_gamma_ineq}
  For $\beta \in \cointerval{\frac{1}{2}, 1}$ and $k \in \N$, $k \ge 2$: $Q^{-1}(k, \beta) \le \frac{k}{k+1}Q^{-1}(k+1, \beta)$.
\end{restatable}

Summing up Lemmas~\ref{lem:confident_mean},~\ref{lem:minor_mean},~\ref{lem:standard_mean},~\ref{lem:inv_gamma_ineq} and using the expression for $\lambda_{k}$ (Algorithm~\ref{algo} line~\ref{algo:lambda_k}) we get:

\begin{corollary}\label{cor:mu_min}
  The average expected value of the bins in the subgroups having at least 2 bins is lower bounded by
  \begin{equation}\label{mu_min}
    \mu_{\min} := \frac{1}{2}\min\left(p_{\max}, \mu_{0}, (1-p_{\max})\lambda_{\min} \right)
  \end{equation}
  where
  \begin{equation}\label{eq:lambda_min}
    \begin{aligned}
      \lambda_{\min} &= \frac{1}{\floor{\frac{1}{s_{\max}}} + 1}Q^{-1}\left(\floor{\frac{1}{s_{\max}}} + 1, 1 - \alpha\right)
    \end{aligned}
  \end{equation}
\end{corollary}

\begin{theorem}\label{th:approx-ratio}
  If RPAP packed items $X_{1},\ldots, X_{n}$ to $M$ bins, and the optimal packing uses $OPT$ bins then $M \le C\cdot OPT + k_{\max} - k_{\min} + 3$,
  where $C$ is the (asymptotic) approximation constant and equals
  $C = \frac{1 + \alpha}{(1-\alpha)\mu_{\min}}$.
\end{theorem}
\begin{proof}
  First, let us consider only the items that belong to the subgroups that were packed into at least 2 bins by RPAP. Without loss of generality let us assume that those are the items $X_{1},\ldots, X_{m}$.
  Let us denote the number of bins those items were packed to by $M'$, the number of bins in the optimal packing of those items by $OPT'$, and their total expected value by $S = \sum_{i=1}^{m}p_{i}s_{i}$. Then from the Lemma~\ref{upper_mean} and Corollary~\ref{cor:mu_min}:
  \[\frac{1}{OPT'}S \le \frac{1 + \alpha}{(1-\alpha)},\quad \frac{1}{M'}S \ge \mu_{\min},\quad M' \le \frac{1 + \alpha}{(1-\alpha)\mu_{\min}}OPT' \le C \cdot OPT.\]

  Finally, notice that we divided the items into $k_{\max} - k_{\min} + 3$ subgroups, so $M - M' \le k_{\max} - k_{\min} + 3$, which is a constant and thus $M'$ is asymptotically equivalent to $M$.
\end{proof}

\subsection{Optimization of the approximation ratio}\label{sec:approximation-opt}
Recall that the approximation constant depends on the values of parameters $p_{\max}$ and $s_{\min}$. To minimize the approximation constant $C = \frac{1 + \alpha}{(1-\alpha)\mu_{\min}}$, we need to maximize the formula for $\mu_{\min}$ (Eq.~\ref{mu_min}).
In case of $p_{\max}$, it comes down to solving the equation: $p_{\max} = (1-p_{\max})\lambda_{\min}$, thus the optimal value is
\begin{equation}\label{eq:p_max}
  p_{\max} = \frac{\lambda_{\min}}{1 + \lambda_{\min}}
\end{equation}

To optimize $s_{\min}$, notice that the expression for $\lambda_{\min}$~(\ref{eq:lambda_min})
does not depend on $s_{\min}$, so we can take arbitrarily small $s_{\min}$ so that
\[\mu_{0} = \frac{2\alpha + s_{\min} - \sqrt{s_{\min}^{2} + 4\alpha s_{\min}}}{2\alpha} \xrightarrow{s_{\min} \to 0} 1.\]
In particular, it is enough to take $s_{\min}$ small enough to make $\mu_{0} \ge p_{\max}$. Solving this inequality for $s_{\min}$ results in:
\begin{equation}\label{eq:s_min}
  s_{\min} \le \frac{\alpha{(1 - p_{\max})}^{2}}{p_{\max}}.
\end{equation}
After such optimizations, we get the approximation constant:
\begin{equation}\label{eq:approx_constant}
  C = 2\frac{1+\alpha}{1 - \alpha}\frac{1 + \lambda_{\min}}{\lambda_{\min}}
\end{equation}
We recall that $C$ depends only on $\alpha$, as $\lambda_{\min}$ (Eq. \ref{eq:lambda_min}) depends on $\alpha$ and $s_{\max}$, but $s_{\max} \leq 1$. In the general case with $s_{\max} = 1$, we get 
\begin{equation}\label{eq:approx_constant_general}
  C = 2 \frac{1+\alpha}{1 - \alpha}\left(1 + \frac{2}{Q^{-1}(2, 1 - \alpha)}\right).
\end{equation}
  
Values of $C$ vary considerably depending on the values of $\alpha$ and $s_{\max}$.
For example for $\alpha = 0.1$ and $s_{\max} = 0.25$: $C \approx 7.47$,
for $\alpha = 0.01$ and $s_{\max} = 1$: $C \approx 29.52$,
while for $s_{\max} = 1$ and $\alpha = 0.001$: $C \approx 90.29$.

To investigate asymptotics of $C$ as $\alpha \to 0$, we need to investigate the asymptotic behavior of $\frac{1}{\lambda_{\min}}$.
Expanding $Q^{-1}(a, z)$ near $z = 1$ with
$Q^{-1}(a, z) = {(-(z - 1)\Gamma(a + 1))}^{1/a} + \O({(z-1)}^{2/a})$~\cite{wolfram},
\begin{equation*}
  C \sim \frac{1}{\lambda_{\min}} \sim \frac{1}{Q^{-1}\left(\floor{\frac{1}{s_{\max}}} + 1, 1-\alpha \right)} = \O\left(\alpha^{-\frac{1}{\floor{\frac{1}{s_{\max}}} + 1}}\right) = \O\left(\sqrt[\floor{1/s_{\max}} + 1]{1/\alpha}\right)
\end{equation*}
and in the general case with $s_{\max} = 1$, we get $C = \O\left(\sqrt{1/\alpha}\right)$.

The resulting asymptotics is worse than Kleinberg's~\cite{bursty} $\O\left(\sqrt{\frac{\log (1/\alpha)}{\log\log (1/\alpha)}}\right)$.
However, we recall that the asymptotic analysis in~\cite{bursty} hides the multiplicative constant arising from splitting items into subgroups.
Thus, the exact approximation ratio of their algorithm is better than ours most likely only for very small $\alpha$ values. In cloud computing, the SLOs are usually not greater than 4 nines (corresponding to $\alpha \ge 0.0001$) thus our analysis most likely results in a better approximation constant for $\alpha$ relevant to the field.

\section{Dependence on the maximal overflow probability}\label{sec:ff-alpha}

From the theoretical perspective, the dependence of the approximation constant on $\alpha$ is not perfect, especially when for other distributions, like Poisson or normal,
there are approximation algorithms whose constant does not depend on $\alpha$~\cite{Breitgand_2012,Cohen_2017} (for the Bernoulli distribution no such algorithm is known).
\cite{bursty} prove that any algorithm
based on a single \emph{effective size}
cannot achieve a better approximation constant than $\Omega(\alpha^{-1/2})$.
The following theorem shows that the family of Any-Fit algorithms (not using the effective size approach) has the same upper bound. We start with a technical lemma (proof in~\cite{paper-appendix}).

\begin{restatable}{lemma}{counterexamplelem}\label{lem:counterexample}
  If $X_{1}, X_{2},\ldots, X_{n} \sim \Ber(2\alpha)$ independent, then there exists $n = \Omega(\alpha^{-1/2})$ for which $\P(\sum_{i=1}^{n} X_{i} \le 1) \ge 1 - \alpha$.
\end{restatable}

\begin{theorem}
Every Any-Fit for scaled Bernoulli trials has $\Omega(\alpha^{-1/2})$ approximation ratio.
\end{theorem}
\begin{proof}
  Let us fix the value of $\alpha' = \alpha + \epsilon$ for an arbitrarily small $\epsilon$. We consider the variables $X_{i} \sim \Ber(\alpha', 1 - \epsilon_{2i - 1})$, $Y_{i} \sim \Ber(1, \epsilon_{2i})$,
  where
  $\epsilon_{1} = \sqrt{\alpha}$, $\epsilon_{i} = \epsilon_{i-1} - \frac{\sqrt{\alpha}}{2^{i+2}}$.
  Then, we can pack $\Omega(\alpha^{-1/2})$ items $Y_{i}$ to a single bin, as taking the first $\alpha^{-1/2}$ of such items requires the capacity of at most $\alpha^{-1/2} \cdot \sqrt{\alpha} = 1$. Similarly, from the Lemma~\ref{lem:counterexample} we can pack $\Omega(\alpha^{-1/2})$ items $X_{i}$ to a single bin.

  Let us now consider the case when those items appear on the input in the order $X_{1}, Y_{1}, X_{2}, Y_{2},\ldots$.
  Using induction, we prove that every Any-Fit algorithm packs into the $i$-th bin just two items, $X_{i}, Y_{i}$.
  First, for every $i$, $X_{i}$ and $Y_{i}$ fit into one bin, because $1 - \epsilon_{2i-1} + \epsilon_{2i} < 1$.
  It is easy to calculate that no further item fits into the bin with items $X_{i}$ and $Y_{i}$.
\end{proof}

\section{Evaluation by simulations}\label{sec:experiments}


We study RPAP and derived algorithms in the average case by comparing it with our modification of FF: First Fit Rounded (FFR), and our implementation of Kleinberg et al.~\cite{bursty}.

FFR rounds up items' sizes to the integer multiple of $\epsilon > 0$, i.e. finds the smallest $k$ such that $\hat s:= k\epsilon \ge s$ and packs the rounded items using FF.
We cannot compare RPAP with First Fit, because First Fit computes overflow probabilities, which is \#P-hard for sums of scaled Bernoulli variables~\cite{bursty}.
Yet, for instances with few very small items (i.e.\ comparable to $\epsilon$), the results of FFR should be close to the results of FF\@. 
In our experiments, we take $\epsilon = 10^{-4}$.

We do not suspect RPAP to perform well compared to FFR, as RPAP was designed for theoretical purposes, so we also analyze a derived algorithm, the RPAP Combined (RPAPC).
RPAPC separates the items into the same groups as RPAP, 
but within each group, it packs an item into the first bin according to  RPAP or FFR\@.
RPAPC in principle works like RPAP with First-Fit as an Any-Fit algorithm. However, when an item does not fit into a bin, RPAPC, just like FFR, approximates the probability of the overflow by rounding up items' sizes to the integer multiple of $\epsilon$.
This allows us to pack some items that RPAP would not pack because of the upper bounds we made for the proofs.
Notice, that this means that RPAPC opens a new bin only if the current item does not fit into any already open bin in its group, so the Lemma~\ref{lem_greedy} holds and from it follows the rest of the Section~\ref{sec:approximation-proof} and the following corollary:

\begin{corollary}\label{cor:rpapc}
  RPAPC is an approximation algorithm with an approximation constant not greater than RPAP's.
\end{corollary}

We also compare our algorithms against Kleinberg et.al.~\cite{bursty} and against its combined version (KleinbergC) --- designed analogically as RPAPC.

\begin{figure*}[t]
  \centering
  \begin{subfigure}[t]{.30\textwidth}
    \includegraphics[width=\textwidth]{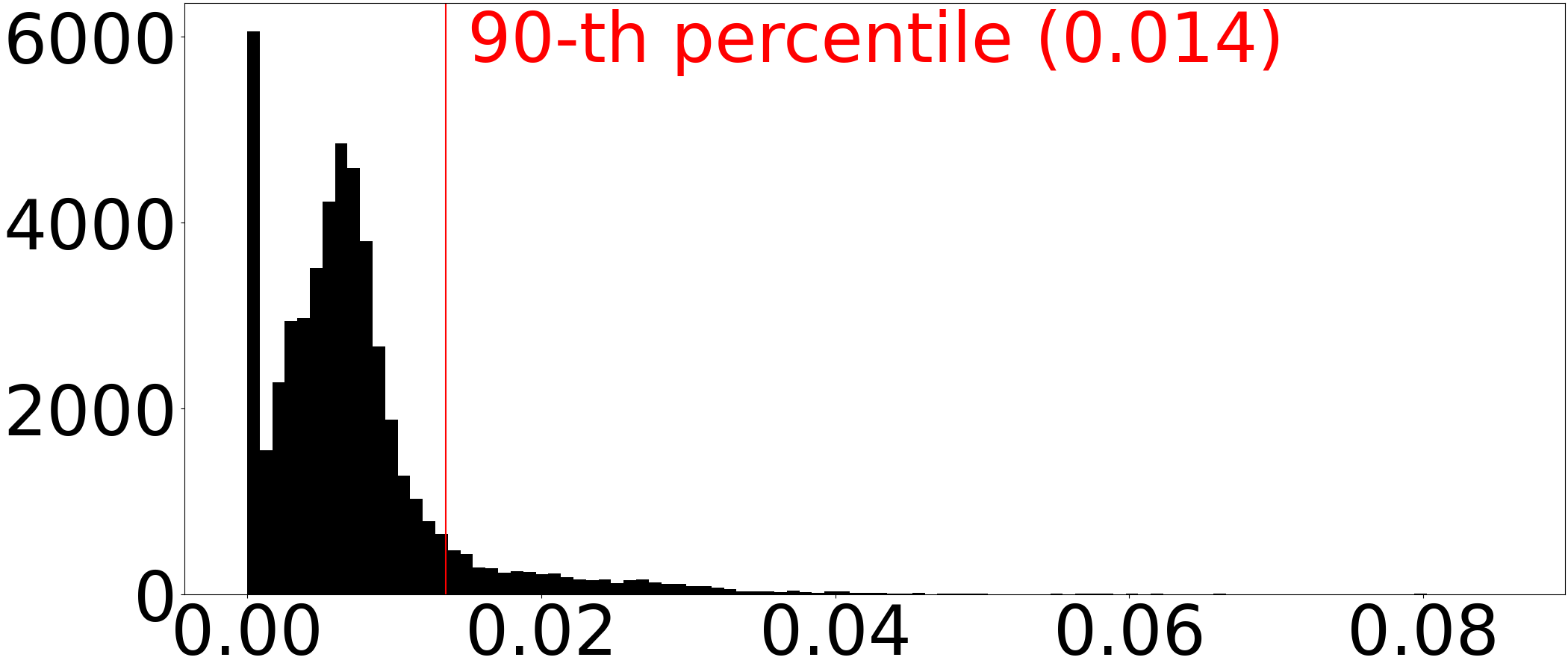}
    \caption{$L_1$ distances}\label{google_scores}
  \end{subfigure}\hspace{.03\textwidth}%
  \begin{subfigure}[t]{.30\textwidth}
    \includegraphics[width=\textwidth]{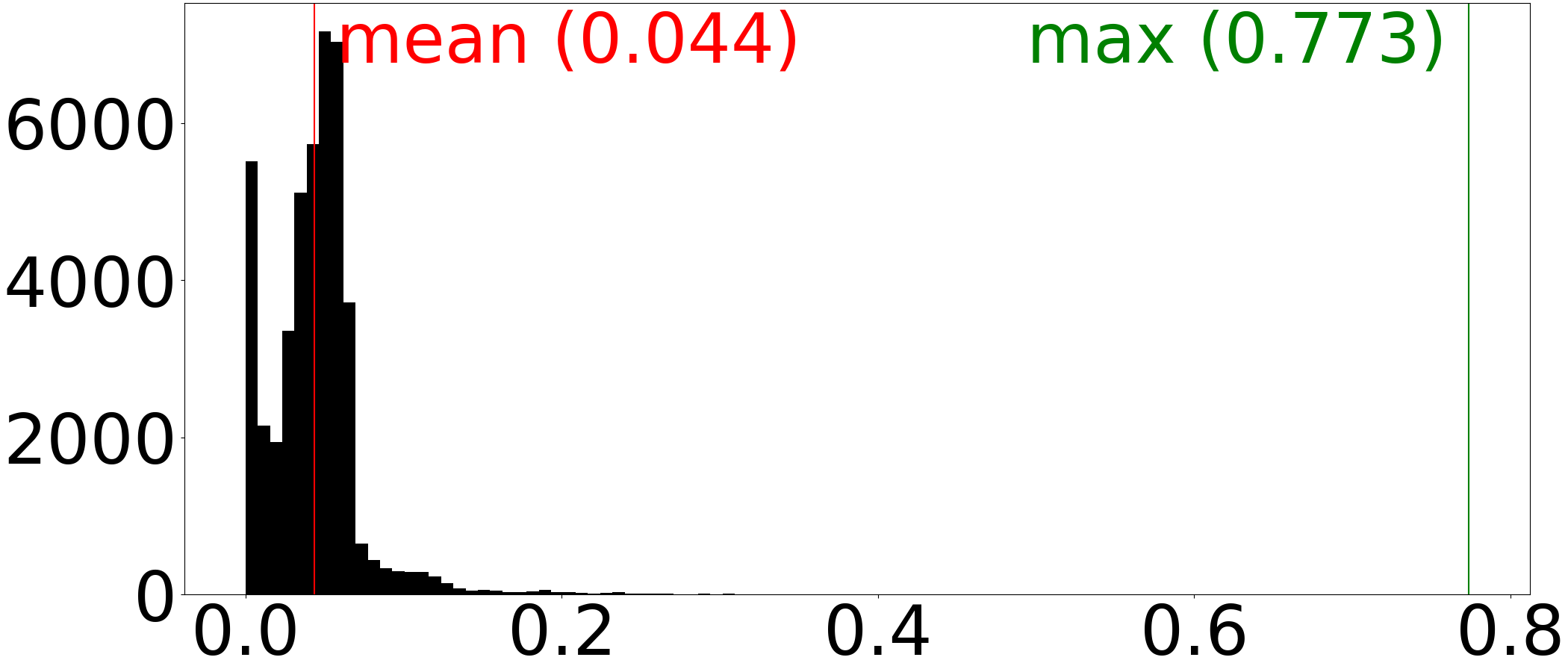}
    \caption{$s_i$, item sizes}\label{google_sizes}
  \end{subfigure}\hspace{.03\textwidth}%
  \begin{subfigure}[t]{.30\textwidth}
    \includegraphics[width=\textwidth]{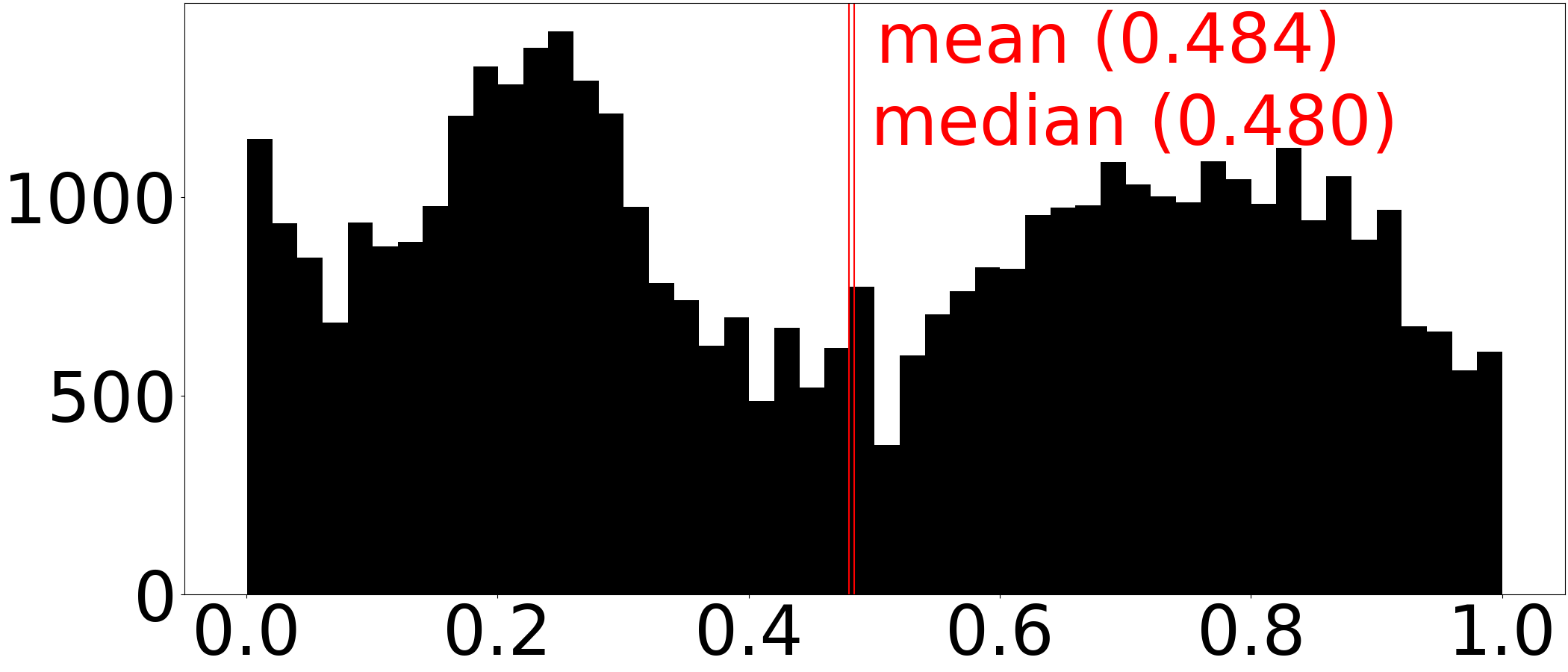}
    \caption{$p_i$, item probabilities}\label{google_probs}
  \end{subfigure}
  \caption{PDFs of items' metrics in \emph{Google} dataset. (a) is a PDF of pair-wise $L_1$ distances between the CDFs of the derived Bernoulli items and the original data. }\label{google_data}
\end{figure*}


\begin{figure*}[t]
  \centering
  \begin{subfigure}{.5\textwidth}
    \includegraphics[width=\textwidth]{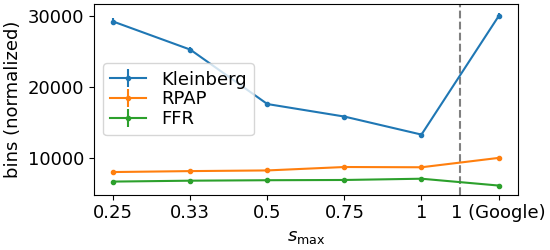}\caption{$\alpha=0.1$}
  \end{subfigure}%
  \begin{subfigure}{.5\textwidth}
    \includegraphics[width=\textwidth]{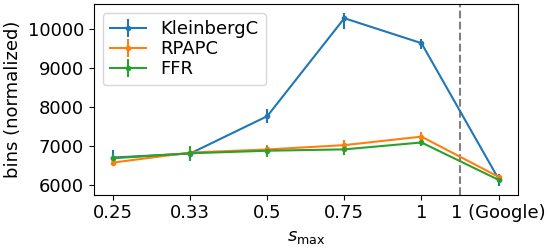}\caption{$\alpha=0.1$} 
  \end{subfigure}%

  \begin{subfigure}{.5\textwidth}
    \includegraphics[width=\textwidth]{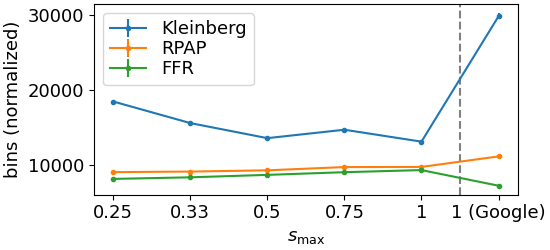}\caption{$\alpha=0.01$}
  \end{subfigure}%
  \begin{subfigure}{.5\textwidth}
    \includegraphics[width=\textwidth]{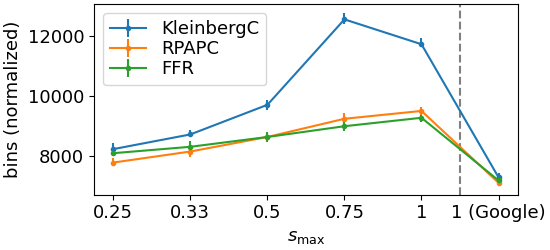}\caption{$\alpha=0.01$}
  \end{subfigure}%

  \begin{subfigure}{.5\textwidth}
    \includegraphics[width=\textwidth]{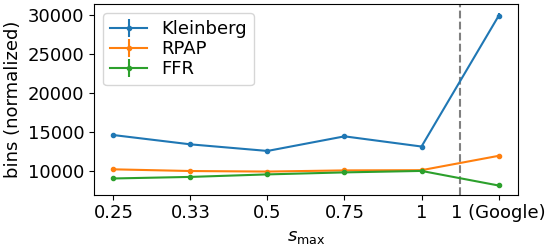}\caption{$\alpha=0.001$}
  \end{subfigure}%
  \begin{subfigure}{.5\textwidth}
    \includegraphics[width=\textwidth]{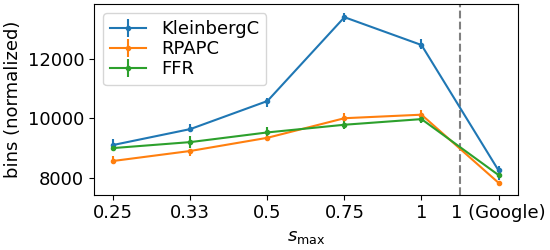}\caption{$\alpha=0.001$}
  \end{subfigure}%

  \caption{Number of bins the items were packed into for \emph{Uniform} (the first 5 columns) and \emph{Google} (last column) instances, by the original (left) and combined (right) algorithms.
  Each point is a median of 10 instances.
  Bars show a minimum and maximum of 10 instances.
  }\label{fig:results}
\end{figure*}

As cloud schedulers work with large volumes of tasks, we are interested in analyzing performance on large instances: we performed experiments with $n=5000$ items.
The values of parameters $s_{\min}$ and $p_{\max}$ for RPAP were set according to the formulas~(\ref{eq:p_max}) and~(\ref{eq:s_min}) to show how RPAP performs without optimizing for a particular dataset. Kleinberg's algorithm does not have any parameters that could be optimized.
For RPAPC we tuned $s_{\min}$ and $p_{\max}$ with respect to $\alpha$, $s_{\max}$ and a dataset (as cloud workload is generally stable over time, schedulers are typically tuned for the standard workload).
To tune, we created new (smaller) instances for every combination of $\alpha$ and $s_{\max}$ and then did a grid search. 

We tested $\alpha \in \{0.1, 0.01, 0.001\}$ and $s_{\max} \in \{1, 0.75, 0.5, 0.33, 0.25\}$. Each dot on a plot is a median from 10 experiments (each having the same parameter values but a different instance), then normalized by dividing by the average expected value of an item in the dataset.
The (very short) vertical lines are the minimum and maximum of those 10 experiments 
--- results are stable over instances.
The lines connecting the dots are introduced for readability.

We generated 3 datasets.
The \emph{uniform} dataset is sampled from the uniform distribution: the sizes are sampled from the $\ocinterval{0,s_{\max}}$ interval, and the probabilities from $\ocinterval{0,1}$.
In the \emph{normal} dataset, sizes are sampled from the normal distribution $N(0.1, 1)$, truncated to the $\ocinterval{0,s_{\max}}$ interval, while the probabilities are sampled from the uniform distribution on the $\ocinterval{0,1}$ interval. Its results were very similar to the results for the uniform dataset, so we omit them.
The \emph{Google} dataset is derived from~\cite{Wilkes2020a}. 
We started with the instantaneous CPU usage dataset~\cite{Janus_2017}.
For every task, we calculated the scaled Bernoulli distribution that is the closest to the task's empirical CPU usage (as measured by the $L_{1}$ distance between CDFs).
Finally, we filtered out the 10\% of items that had the highest distance from the original data (Figure~\ref{google_scores}), in order not to experiment on tasks for which the Bernoulli approximation is the least exact. 
As on \emph{Google}, a vast majority of items is small (Figure~\ref{google_sizes}), results for different $s_{\max}$ are very similar, we only show $s_{\max}=1$.
Figure~\ref{fig:results} shows results.

As expected, both RPAP and Kleinberg algorithms produce significantly worse results than FFR for all datasets (with \emph{Google} being particularly unsatisfactory). In contrast, RPAPC achieves even over $4\%$ better results than FFR on the \emph{Google} dataset for small $\alpha$ values, and on the \emph{uniform} dataset for small $s_{\max}$ values. Moreover, the overflow probabilities in bins packed by RPAPC are on average lower than those in bins packed by FFR (Figure~3 in the appendix~\cite{paper-appendix}): a good packing algorithm can result in both a lower number of bins and a lower overflow probability. On the other side, the KleinbergC algorithm performs worse than both FFR and RPAPC for all datasets, although for \emph{Google} the difference is small. There are significant differences between results on the \emph{uniform} and the \emph{Google} datasets. A possible reason is that the \emph{Google} dataset has a skewed distribution of sizes of items (Figure~\ref{google_sizes}) with mean $0.044$ and maximal size $0.773$, although probabilities are distributed reasonably uniformly (Figure~\ref{google_probs}). 

\section{Conclusions}

We propose RPAP, an online algorithm for Stochastic Bin Packing with scaled Bernoulli items. 
RPAP produces a viable packing that keeps the overflow probability of any bin below the requested parameter $\alpha$.
We also prove that RPAP is an approximation algorithm with an approximation factor that depends only on the maximal overflow probability $\alpha$.
%
We derive a combined approach, RPAPC, that has the same guarantees as RPAP.
In simulations, we compare RPAP and RPAPC with \cite{bursty}, a state-of-the-art algorithm with proven worst-case guarantees, and FFR, a heuristic with no guarantees. Our approaches consistently surpass \cite{bursty}. 
Additionally, RPAPC is close to FFR, outperforming it 
on \emph{Google} by 4\% and on \emph{Uniform} datasets with small items.

\noindent \textbf{Acknowledgements and Data Availability:} Authors thank Arif Merchant for his comments on the manuscript. This research is supported by a Polish National Science Center grant Opus (UMO-2017/25/B/ST6/00116). Data supporting this study and the code used to perform the simulations and generate the plots will be available from~\cite{artifact}.

\bibliographystyle{splncs04}
\bibliography{bibliography}

\newpage{}
\section{Appendix}

\begin{figure*}[t]
  \centering
  \begin{subfigure}{.5\textwidth}
    \includegraphics[width=\textwidth]{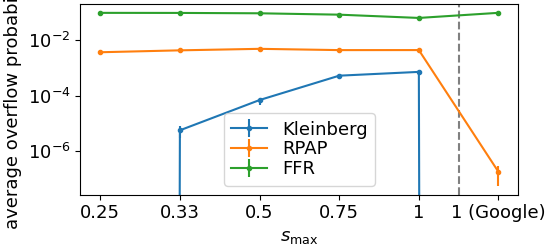}\caption{$\alpha=0.1$}
  \end{subfigure}%
  \begin{subfigure}{.5\textwidth}
    \includegraphics[width=\textwidth]{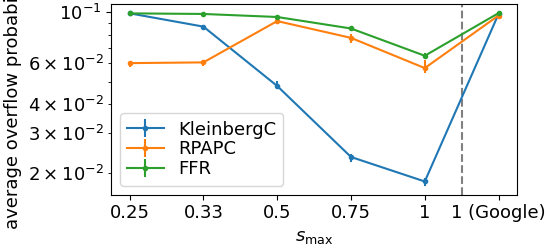}\caption{$\alpha=0.1$}
  \end{subfigure}%

  \begin{subfigure}{.5\textwidth}
    \includegraphics[width=\textwidth]{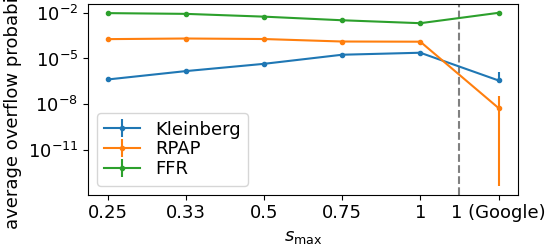}\caption{$\alpha=0.01$}
  \end{subfigure}%
  \begin{subfigure}{.5\textwidth}
    \includegraphics[width=\textwidth]{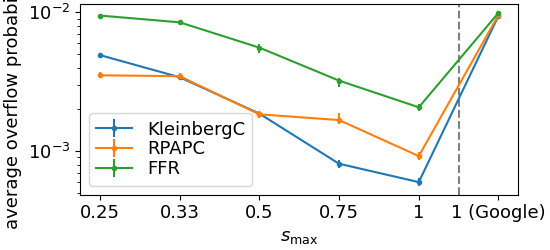}\caption{$\alpha=0.01$}
  \end{subfigure}%

  \begin{subfigure}{.5\textwidth}
    \includegraphics[width=\textwidth]{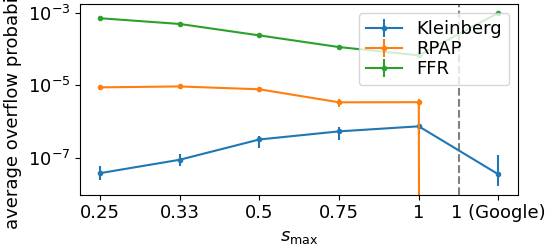}\caption{$\alpha=0.001$}
  \end{subfigure}%
  \begin{subfigure}{.5\textwidth}
    \includegraphics[width=\textwidth]{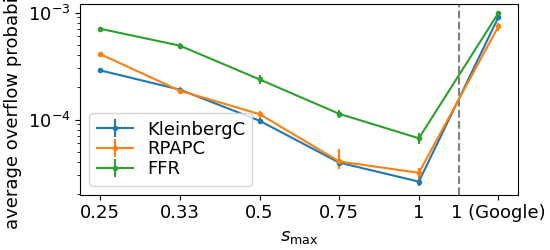}\caption{$\alpha=0.001$}
  \end{subfigure}%

  \caption{Average bin overflow probability for \emph{Uniform} (the first 5 columns) and \emph{Google} (last column) instances, by the original (left) and combined (right) algorithms. Each point is a median of 10 instances. Bars show a minimum and maximum of 10 instances. }\label{fig:average_overflow}
\end{figure*}

\berpoilemma*
\begin{proof}
  For $t \ge 1$ we have $F_{X}(t) = 1 > F_{Y}(t)$, for $t = 0$:
  \[F_{X}(0) = 1-p = e^{-\ln\left(\frac{1}{1-p}\right)} \ge e^{-\lambda} = F_{Y}(0).\]
\end{proof}

\sumtwolemma*
\begin{proof}
  \begin{equation*}
    \begin{split}
      F_{X_{1}+X_{2}}(t) &= \sum_{x} \P(X_{1} = x)\P(X_{2} \le t - x) \ge \\
      &\ge \sum_{x_{1}}\sum_{x_{2} \ge x_{1}} \P(X_{1} = x_{1})\P(Y_{2} = t - x_{2}) \text{.}
    \end{split}
  \end{equation*}
  This sum has countably many nonzero components, and all of them are positive, so it is absolutely convergent, and we can change the order of summation to get:
  \begin{equation*}
    \begin{split}
      F_{X_{1}+X_{2}}(t) &\geq \sum_{x_{2}} \P(Y_{2} = t - x_{2}) \sum_{x_{1} \le x_{2}} \P(X_{1} = x_{1}) \ge \\
      &\ge \sum_{x}\P(Y_{2} = t - x)F_{Y_{1}}(x) = F_{Y_{1}+Y_{2}}(t) \text{.}
    \end{split}
  \end{equation*}
\end{proof}

\stochmajorizationlem*
\begin{proof}
  From Lemma~\ref{lem_bern_poiss} $\forall_{i}\forall_{t} F_{X_i}(t) \geq F_{P_{i}}(t)$, so from Lemma~\ref{lem_cdf_sum} $\forall_{t} F_{B}(t) \geq F_{P}(t)$ and
  \[\P(B > 1) = 1 - F_{B}(1) \le 1 - F_{P}(1) = \P(P > 1)\text{.}\]
\end{proof}

\invgammaineqlem*
\begin{proof}
From the facts that $Q(k, Q^{-1}(k, \beta)) = \beta$ and that $Q(s, x)$ is decreasing w.r.t. $x$, we get
\[Q^{-1}(k, \beta) = \max\{\lambda > 0:\ Q(k, \lambda) \ge \beta\}\]
so we can rewrite the above inequality as
\[\max\{\lambda > 0:\ Q(k, \lambda) \ge \beta\} \le \frac{k}{k+1}\max\{\lambda > 0:\ Q(k+1, \lambda) \ge
  \beta\}.\]
By substituting $\lambda \to \frac{k+1}{k}\lambda$ on the right side of the inequality we get
\[\max\{\lambda > 0:\ Q(k, \lambda) \ge \beta\} \le \max\left\{\lambda > 0:\ Q\left(k+1, \frac{k+1}{k}\lambda\right) \ge \beta\right\}\]
so it is enough to show that for
\[\lambda \in \left\{Q^{-1}(k, \beta):\ \beta \in \cointerval{\frac{1}{2}, 1}\right\} \subseteq \ocinterval{0, Q^{-1}\left(k, \frac{1}{2}\right)}\]
we have
\[Q(k,\lambda) \le Q\left(k+1, \frac{k+1}{k}\lambda\right).\]
By rewriting the above inequality using the definition $Q(s, x) = \frac{\Gamma(s,x)}{\Gamma(s)}$, and the equation
\[\Gamma(s+1, x) = s\Gamma(s, x) + x^{s}e^{-x}\]
we get
\[\frac{\Gamma(k, \lambda)}{\Gamma(k)} \le \frac{k\Gamma(k, \frac{k+1}{k}\lambda) + {\left(\frac{k+1}{k}\lambda\right)}^{k} e^{-\frac{k+1}{k}\lambda} }{\Gamma(k+1)}\]
Now, using the definition of $\Gamma(s,x)$ and the equation $\Gamma(k) = (k-1)!$ we arrive at
\[\int_{\lambda}^{\frac{k+1}{k}\lambda}t^{k-1}e^{-t}dt \le \frac{1}{k}{\left(\frac{k+1}{k}\lambda\right)}^{k} e^{-\frac{k+1}{k}\lambda}\]
Let us notice that if the function $t^{k-1}e^{-t}$ is increasing on the interval $(\lambda, \frac{k+1}{k}\lambda)$, then the proof is finished:
\begin{equation*}
  \begin{aligned}
    \int_{\lambda}^{\frac{k+1}{k}\lambda}t^{k-1}e^{-t}dt &\le \frac{\lambda}{k}{\left(\frac{k+1}{k}\lambda\right)}^{k-1} e^{-\frac{k+1}{k}\lambda} \le \\
    &\le \frac{k+1}{k}\frac{\lambda}{k}{\left(\frac{k+1}{k}\lambda\right)}^{k-1} e^{-\frac{k+1}{k}\lambda}.
  \end{aligned}
\end{equation*}
The derivative $\frac{\partial}{\partial t}(t^{k-1}e^{-t}) = -e^{-t}t^{k-2}(t - k + 1)$ is non-negative for $t \le k - 1$, so we need $\frac{k+1}{k}\lambda \le k - 1$, thus it is enough to show that $Q^{-1}(k, \frac{1}{2}) \le \frac{k(k-1)}{k+1}$, or equivalently $Q(k, \frac{k(k-1)}{k+1}) \ge \frac{1}{2}$, but
\[Q(k, \frac{k(k-1)}{k+1}) \ge Q(k, k-1) = \P(\Poi(k-1) \le k-1)\]
so the thesis follows from the fact that for $n \in \N$ median of $\Poi(n)$ equals $n$~\cite{Adell_2005}.
\end{proof}

\counterexamplelem*
\begin{proof}
  Let us fix $n = c\alpha^{-1/2}$ for some $c > 0$. From the Lemmas~\ref{lem_bern_poiss},~\ref{lem_cdf_sum} we have:
  \begin{equation*}
  \P\left(\sum_{i=1}^{n}X_{i} \le 1\right) \ge \P\left(\Poi\left(n\ln\left(\frac{1}{1 - 2\alpha}\right)\right) \le 1\right) = Q\left(2, c\alpha^{-1/2}\ln\left(\frac{1}{1 - 2\alpha}\right)\right).
  \end{equation*}
  As
  $\ln\left(\frac{1}{1 - 2\alpha}\right) \le 4\alpha$ for $\alpha \le \frac{1}{4}$, $\P\left(\sum_{i=1}^{n}X_{i} \le 1\right) \ge Q(2, 4c\sqrt{\alpha})$.
  Applying $Q^{-1}(2,\cdot)$ to the inequality $Q(2, 4c\sqrt{\alpha}) \ge 1 - \alpha$, we arrive at the inequality $4c\sqrt{\alpha} \le Q^{-1}(2, 1-\alpha)$, and the existence of such constant $c$ follows from the asymptotic expansion of $Q^{-1}$ near 1 in~\cite{wolfram}.

\end{proof}

\end{document}